\documentclass[11pt]{article}
\usepackage[margin=1in]{geometry}
\usepackage{amssymb,amsmath,amsthm}
\usepackage{verbatim,hyperref}

\newtheorem{theorem}{Theorem}
\newtheorem{lemma}{Lemma}
\newtheorem{observation}{Observation}

\theoremstyle{definition}
\newtheorem{definition}{Definition}

\DeclareMathOperator{\sol}{SOL}
\DeclareMathOperator{\alg}{ALG}
\DeclareMathOperator{\polylog}{polylog}
\DeclareMathOperator{\lca}{lca}
\DeclareMathOperator{\sur}{sur}
\DeclareMathOperator{\csur}{csur}

\title{Polylogarithmic Bounds on the Competitiveness of \\
Min-cost (Bipartite) Perfect Matching with Delays}

\author{
Yossi Azar\\
\texttt{azar@tau.ac.il}
\and
Ashish Chiplunkar\\
\texttt{ashish.chiplunkar@gmail.com}
\and
Haim Kaplan\\
\texttt{haimk@post.tau.ac.il}
}

\date{}

\begin{document}
\maketitle

\thispagestyle{empty}
\begin{abstract}
We consider the problem of online Min-cost Perfect Matching with Delays (MPMD) recently introduced by Emek et al, (STOC 2016). This problem is defined on an underlying $n$-point metric space. An adversary presents real-time requests online at points of the metric space, and the algorithm is required to match them, possibly after keeping them waiting for some time. The cost incurred is the sum of the distances between matched pairs of points (the connection cost), and the sum of the waiting times of the requests (the delay cost). We prove the first logarithmic upper bound and the first polylogarithmic lower bound on the randomized competitive ratio of this problem. We present an algorithm with a competitive ratio of $O(\log n)$, which improves the upper bound of $O(\log^2n+\log\Delta)$ of Emek et al, by removing the dependence on $\Delta$, the aspect ratio of the metric space (which can be unbounded as a function of $n$). The core of our algorithm is a deterministic algorithm for MPMD on metrics induced by edge-weighted trees of height $h$, whose cost is guaranteed to be at most $O(1)$ times the connection cost plus $O(h)$ times the delay cost of every feasible solution. The reduction from MPMD on arbitrary metrics to MPMD on trees is achieved using the result on embedding $n$-point metric spaces into distributions over weighted hierarchically separated trees of height $O(\log n)$, with distortion $O(\log n)$. We also prove a lower bound of $\Omega(\sqrt{\log n})$ on the competitive ratio of any randomized algorithm. This is the first lower bound which increases with $n$, and is attained on the metric of $n$ equally spaced points on a line.

The problem of Min-cost Bipartite Perfect Matching with Delays (MBPMD) is the same as MPMD except that every request is either positive or negative, and requests can be matched only if they have opposite polarity. We prove an upper bound of $O(\log n)$ and a lower bound of $\Omega(\log^{1/3}n)$ on the competitive ratio of MBPMD with a more involved analysis.
\end{abstract}


\section{Introduction}

The problem of finding a maximum / minimum weight (perfect) matching in an edge-weighted graph has been one of the central problems in algorithmic graph theory, and has been the topic of extensive research ever since the seminal work by Edmonds \cite{Edmonds_JRNBS65, Edmonds_CJM65}. Given a graph with positive edge weights, a (perfect) matching $M$ is a subset of edges such that no two edges in $M$ have a common endpoint (and every vertex is an endpoint of some edge in $M$), and its weight is the sum of the weights of the edges that it contains. The online version of the matching problem comes in numerous flavors, motivated by applications from a variety of domains. Some of the prominent lines of work, among the plethora of results on online matching, include min-weight perfect matchings with online vertex arrival \cite{KalyanasundaramP_JAlgorithms93, KhullerMV_TheorComputSci94, KoutsoupiasN_WAOA03, MeyersonNP_SODA06, GuptaL_ICALP12, BansalBGN_Algorithmica14}, max-cardinality or max-weight matchings with online vertex arrival \cite{KarpVV_STOC90, FeldmanMMM_FOCS09, BahmaniK_ESA10, AggarwalGKM_SODA11, MahdianY_STOC11, DevanurJ_STOC12, DevanurJK_SODA13}, and max-cardinality or max-weight matchings with online edge arrival \cite{McGregor_Approx05, Varadaraja_ICALP11, EpsteinLSW_STACS13}. We remark that the list of references is merely a tip of the iceberg of literature on online matchings.

A variant of the online matching problem, which gaming platforms such as chess.com, bridgebase.com, etc.\ face, is the following. Players log into the website, and express their desire to participate in multi-player gaming sessions. The platform is required to create \textit{tables} of an appropriate fixed number of players (two in case of chess, and four in case of bridge). The (dis)satisfaction experienced by a player is a combination of the time taken by the platform to assign her a table, and several player-dependent factors such as the differences between their ratings, their average time to move, etc. In order to improve participation, the platform has to run a table-assignment algorithm which attempts to minimize the dissatisfaction of the players. With this motivation, and restricting attention to two-player games, Emek et al.\ \cite{EmekKW_STOC16} defined the problem of Min-cost Perfect Matching with Delays (MPMD), which is general enough to have numerous other applications, such as finding roommates, carpooling, etc.

A close cousin of MPMD is the problem of Min-cost \textit{Bipartite} Perfect Matching with Delays (MBPMD), which is motivated from applications such as organ transplantation, transportation platforms like Uber, etc. Here, each player is of one of two types (say donor or acceptor, in case of organ transplantation), and we are required to pair up players of dissimilar types. The cost of the matching, as before, is determined by the waiting time of the players and the dissimilarity in the pairs.

\textbf{Problem definition (informal):} In the MPMD problem on an underlying $n$-point metric space, each point represents a \textit{type} of players, and the distance between two points, say $p$ and $q$, is the dissatisfaction of players of types $p$ and $q$ if they are paired up. We call this distance the \textit{connection cost}. The online input given to the algorithm is a real-time sequence of players and their types, which we will call requests, where each request is revealed to the algorithm only at its arrival time, when it is unaware of the future requests. The algorithm is required to create pairs of players after possibly keeping them waiting for some time. The objective of the algorithm is to minimize the connection cost plus the total waiting time of the requests, also called the \textit{delay cost}. The MBPMD problem is the same as MPMD, except that each request is either positive (i.e.\ a producer) or negative (i.e.\ a consumer) and the algorithm is required to pair up positive and negative requests. Apparently, there is no reduction between MPMD and MBPMD, though MBPMD appears to be harder, arguably.

We remark that the offline version of MPMD on a metric space $\mathcal{M}$, where the entire input is known in advance, trivially reduces to finding a minimum cost perfect matching in a set of points in the metric space $\mathcal{M}\times\mathbb{R}$, with the distance between two points $(p_1,t_1)$ and $(p_2,t_2)$ being the sum of the distance between $p_1$ and $p_2$ in $\mathcal{M}$, and $|t_1-t_2|$. Similarly, the offline version of MBPMD translates to finding a minimum cost perfect matching between two sets of points in $\mathcal{M}\times\mathbb{R}$, one given by the positive requests and the other by the negative requests.

\textbf{Competitive Analysis:} In a typical online problem, an input is given to an algorithm in pieces, and the algorithm is constrained to make irrevocable decisions while processing every piece. The popular technique used to measure the performance of an online algorithm is competitive analysis \cite{BorodinE}, where we prove bounds on its \textit{competitive ratio}. A (randomized) online algorithm is said to have a competitive ratio of $\alpha$, if on every possible input, the algorithm produces a solution with (expected) cost at most $\alpha$ times the cost of the optimum solution to the instance, plus a constant which is independent of the online input. We assume that the input is generated by an adversary, who knows the algorithm and can force it to incur a large cost, while the adversary itself is able to produce a much cheaper solution (with full knowledge of the future input). When the algorithm is randomized, we assume that the adversary is \textit{oblivious}, that is, it does not have access to the random choices made by the algorithm.

\textbf{Background:} Emek et al.\ \cite{EmekKW_STOC16} gave the first online algorithm for MPMD with a finite competitive ratio. Given an $n$-point metric space $\mathcal{M}$ with aspect ratio $\Delta$ (the ratio of the maximum distance to the minimum distance), they consider its embedding into a distribution over metrics given by hierarchically separated full binary trees, with distortion $O(\log n)$. They then give a randomized algorithm for hierarchically separated trees, and they bound its competitive ratio using an appropriately defined stochastic process which captures the behavior of the algorithm. This results in an algorithm for the original metric $\mathcal{M}$ having competitive ratio $O(\log^2n+\log\Delta)$. Emek et al.\ remark that a constant lower bound on the competitive ratio exists even on two-point metrics, since MPMD captures the ski-rental problem.

\textbf{Our contributions:} The results of Emek et al.\ \cite{EmekKW_STOC16} naturally raise the questions: whether the competitive ratio depends on the number of points, and whether it depends on the aspect ratio. We answer both of these questions in this paper. On the one hand, we prove that the competitive ratio can be made independent of the aspect ratio.  On the other hand, we also prove that the competitive ratio must depend on the number of points. Our contributions can be summarized as follows.
\begin{enumerate}
\item Deterministic $O(h)$-competitive algorithms for MPMD and MBPMD on metrics given by trees of height $h$. In particular, these are deterministic $O(1)$-competitive algorithms for uniform metrics. The algorithm of Emek et al.\ is randomized, and has an $\Omega(\log n)$ competitive ratio even on uniform metrics.
\item $O(\log n)$-competitive algorithms for MPMD and MBPMD on arbitrary $n$-point metrics. This improves the bound of $O(\log^2n+\log\Delta)$ by Emek et al, by removing dependence on $\Delta$, the aspect ratio of the metric space, which can potentially be unbounded as a function of $n$.
\item Lower-bound constructions which prove that on $n$-point metrics, the competitive ratio of any randomized algorithm for MPMD must be $\Omega(\sqrt{\log n})$, and the competitive ratio of any randomized algorithm for for MBPMD must be $\Omega(\log^{1/3}n)$. These are the first lower bounds which increase with $n$, and the former confirms the conjecture by Emek et al. 
\end{enumerate}
Our deterministic algorithm for MPMD (resp.\ MBPMD) on trees is a simple algorithm which maintains one timer $z_u$ (resp.\ two timers $z^+_u$ and $z^-_u$) for every vertex $u$ of the tree, that measures the amount of time for which the subtree rooted at $u$ had an odd number of pending requests (resp.\ the time-integrals of the ``surplus'' and the ``deficiency'' in the subtree rooted at $u$). We reduce M(B)PMD on arbitrary metrics to M(B)PMD on trees by using the technique of randomized embedding \cite{FakcharoenpholRT_JCSS04} followed by a height reduction step \cite{BansalBMN_JACM15}. Although the technique guarantees an embedding into a (weighted) hierarchically separated tree, we do not need the hierarchical separation condition; but only that the height of the tree is small. It is believable that the competitive ratio of any algorithm which uses an embedding technique must be bounded from below by the distortion of the embedding. Our algorithm is, therefore, an optimal tree-embedding based algorithm, since it is known that there exist metric spaces which do not embed into distribution over tree metrics with distortion $o(\log n)$ (Theorem 9 of \cite{Bartal_FOCS96}).

In contrast to Emek et al, we do not need to embed metrics into binary trees, and the height of our trees is $O(\log n)$, independent of the aspect ratio $\Delta$. Moreover, our algorithm for MPMD on tree metrics is deterministic and has a relatively simple proof of competitiveness, whereas the previous algorithm is randomized and has a fairly involved analysis.

Our lower bounds are achieved on the metric space of $n$ equally spaced points in the unit interval. We invoke Yao's min-max technique \cite{BorodinE_InfComput99, StougieV_OperResLett02, Yao_FOCS77} and give a probability distribution over inputs, which defeats every deterministic online algorithm by a factor of $\Omega(\sqrt{\log n})$ in case of MPMD, and $\Omega(\log^{1/3}n)$ in case of MBPMD. 

\textbf{Extensions:} Emek et al.\ \cite{EmekKW_STOC16} also analyze a variant, called MPMDfp, where requests can be cleared at a fixed cost. They give a reduction from MPMDfp on a metric space $\mathcal{M}$ to MPMD on an appropriately defined metric space containing two copies of $\mathcal{M}$, and show that this only introduces a factor of $2$ in the competitive ratio. The same reduction, along with our algorithm for MPMD, results in an $O(\log n)$ competitive algorithm for MPMDfp.

\textbf{Organization of the paper:} We first define the problems and the related terminology formally in Section \ref{sec_prelim}. Section \ref{sec_ub} is dedicated to proving the upper bounds, where we first state the embedding result and show how the distortion of the embedding and the competitive ratio on tree metrics translates to competitive ratio on arbitrary metrics. We then follow it up by our algorithms for MPMD and MBPMD on tree metrics. We prove the lower bound results in Section \ref{sec_lb}. We conclude by stating a few remarks and related open problems in Section \ref{sec_rem}.

\section{Preliminaries}\label{sec_prelim}

A \textit{metric space} $\mathcal{M}$ is a set $S$ equipped with a distance function $d:S\times S\longrightarrow\mathbb{R}^+$ such that $d(x,x)=0$ for all $x\in S$, $d(x,y)=d(y,x)$ for all $x,y\in S$, and $d(x,y)+d(y,z)\geq d(x,z)$ for all $x,y,z\in S$. The online problem of Min-cost Perfect Matching with Delays (MPMD) on a finite metric space $\mathcal{M}=(S,d)$, as defined in \cite{EmekKW_STOC16}, is the following. The metric space is an offline input to the algorithm. An online input instance $I$ over $S$ is a sequence of requests $\langle(p_i,t_i)\rangle_{i=1}^{m}$, where $m$ is even, each $p_i\in S$, and $t_1\leq t_2\leq\cdots\leq t_m$. The request $(p_i,t_i)$ is revealed at time $t_i$. The algorithm is required to output a perfect matching of requests in real time. For each pair $(i,j)$ of requests output by the algorithm at time $t$ (where $t\geq\max(t_i,t_j)$), the algorithm pays a connection cost of $d(p_i,p_j)$ and a delay cost of $(t-t_i)+(t-t_j)$. The offline connection cost of creating the pair $(i,j)$ is $d(p_i,p_j)$, and the offline delay cost is $|t_i-t_j|$. The offline cost of a perfect matching on $\{1,\ldots,m\}$ is the total connection cost and delay cost over all pairs in the matching. The optimal solution is a perfect matching with the minimum offline cost. In the problem of Min-cost Bipartite Perfect Matching with Delays (MBPMD), the $i^{\text{\tiny{th}}}$ request is $(p_i,b_i,t_i)$, where $p_i\in S$, $b_i\in\{+1,-1\}$, and $t_i$ is the arrival time. The algorithm is allowed to output the pair $(i,j)$ only if $b_ib_j=-1$, and incurs the same cost as in MPMD.

Although we quantify the performance of online algorithms by their competitive ratio, we need to define a more general notion of competitiveness, customized for M(B)PMD, for stating our intermediate results. Given 
an instance $I$ of M(B)PMD
and an arbitrary solution $\sol$ of $I$, let $\sol_d$ denote its connection cost with respect to the metric $d$, $\sol_t$ denote its delay cost, and (with a slight abuse of notation) $\sol$ denote its total cost. Given a randomized algorithm $\mathcal{A}$, let $\mathcal{A}(I)$ be the random variable denoting the algorithm's total cost on $I$.

\begin{definition}
A randomized online algorithm $\mathcal{A}$ for M(B)PMD on a metric space $\mathcal{M}=(S,d)$ is said to be $\alpha$-\textit{competitive} if for every instance $I$ on $S$ and every solution $\sol$ of $I$, $\mathbb{E}[\mathcal{A}(I)]\leq\alpha\times\sol$. More generally, the algorithm is said to be $(\beta,\gamma)$-\textit{competitive} if for every instance $I$ on $S$ and every solution $\sol$ of $I$, $\mathbb{E}[\mathcal{A}(I)]\leq\beta\times\sol_d+\gamma\times\sol_t$.\footnote{To prove that an algorithm is $\alpha$-competitive, it suffices to compare its cost with the cost of the optimum solution. In contrast, to prove that the algorithm is $(\beta,\gamma)$-competitive, we have to compare its performance with that of every solution.}
\end{definition}

Note that an $\alpha$-competitive algorithm is trivially $(\alpha,\alpha)$-competitive, and a $(\beta,\gamma)$-competitive algorithm is trivially $(\max(\beta,\gamma))$-competitive.

A key ingredient in our algorithm is the technique of embedding metrics into distributions over tree metrics with low distortion. We define these notions formally.

\begin{definition}
Let $\mathcal{M}=(S,d)$ be a finite metric space, and let $\mathcal{D}$ be a probability distribution over metrics on a finite set $S'$. We say that $\mathcal{M}$ \textit{embeds} into $\mathcal{D}$ if $S\subseteq S'$, and for every $x,y\in S$ and every metric space $\mathcal{M}'=(S',d')$ in the support of $\mathcal{D}$, we have $d(x,y)\leq d'(x,y)$. The \textit{distortion} of this embedding is defined to be
\[\mu=\max_{x,y\in S\text{, }x\neq y}\frac{\mathbb{E}_{\mathcal{M'}=(S',d')\sim\mathcal{D}}[d'(x,y)]}{d(x,y)}\]
\end{definition}

\section{The $O(\log n)$ Upper bound}\label{sec_ub}

Our focus in this section is to give algorithms for MPMD and MBPMD on arbitrary metrics, and thus, to prove the following result.

\begin{theorem}\label{thm_ub}
There exist randomized online algorithms with a competitive ratio of $O(\log n)$ for MPMD and MBPMD on $n$-point metric spaces.
\end{theorem}

Analogous to the algorithm by Emek et al, our algorithms also exploit results on embedding arbitrary metrics into distributions over tree metrics.
In the subsequent subsections, we first show how the competitive ratio of an algorithm on an embedding metric space (tree metrics) translates into its competitive ratio on the embedded space. We then proceed to state the algorithm on tree metrics, and bound its competitive ratio.

\subsection{Reduction to Tree Metrics}\label{subsec_reduction}

The reduction of M(B)PMD on arbitrary metrics to M(B)PMD on tree metrics is achieved by the celebrated result of Fakcharoenphol et al.\ \cite{FakcharoenpholRT_JCSS04}, which gives an embedding of an arbitrary metric space into a distribution over hierarchically separated trees (HSTs) \cite{Bartal_FOCS96}. Informally, a $\sigma$-HST over a set $S$ has $S$ as its set of leaves, and the distance between any two points in $S$, under the HST metric, is determined by the level of their lowest common ancestor (LCA), with the root defined to be at the highest level. If the LCA is at a level $l$, then the distance is $d_l$, where $d_l\geq\sigma d_{l-1}$ ($\sigma$-hierarchical separation). Emek et al.\ used an embedding of HSTs into hierarchically separated binary trees (HSBTs), and design a randomized algorithm for HSBTs. Instead of this, we use a result by Bansal et al.\ \cite{BansalBMN_JACM15} to reduce the height of the tree to $O(\log n)$, since we design an algorithm for MPMD on tree metrics, with competitive ratio depending only on the height of the tree. While the height reduction step might lose the hierarchical separation, this is not a concern, since our algorithm works on arbitrary tree metrics. The overall embedding result that we need is stated in the following lemma, whose proof is deferred to Appendix \ref{app_a}.

\begin{lemma}\label{lem_embed}
Any $n$-point metric space $\mathcal{M}$ can be embedded, with distortion $O(\log n)$, into a distribution $\mathcal{D}$, supported on metrics induced by trees of height $O(\log n)$.
\end{lemma}

The other ingredient of the reduction is the following result, which states how the competitive ratio on the embedding metric is translated into the competitive ratio on the embedded metric. Although this was proved and used by Emek et al, we reproduce its proof in Appendix \ref{app_a}, for the sake of completeness.

\begin{lemma}\label{lem_reduction}
Suppose that a metric space $\mathcal{M}=(S,d)$ can be embedded into a distribution $\mathcal{D}$ supported on metric spaces over a set $S'\supseteq S$ with distortion $\mu$. Additionally, suppose that for every metric space $\mathcal{M'}$ in the support of $\mathcal{D}$, there is a deterministic online $(\beta,\gamma)$-competitive algorithm $\mathcal{A}^{\mathcal{M}'}$ for M(B)PMD on $\mathcal{M}'$. Then there is a $(\mu\beta,\gamma)$-competitive (and thus, $(\max(\mu\beta,\gamma))$-competitive) algorithm $\mathcal{A}$ for M(B)PMD on $\mathcal{M}$.
\end{lemma}

Given an $n$-point metric space, we have an embedding into distribution over tree metrics of height $h=O(\log n)$ with distortion $\mu=O(\log n)$, resulting from Lemma \ref{lem_embed}. In the next two subsections, we prove that there exist $(O(1),O(h))$-competitive algorithms for MPMD and MBPMD on tree metrics of height $h$, i.e.\ the algorithms always give a solution whose cost is at most $O(1)$ times the connection cost plus $O(h)$ times the delay cost of any solution. As a consequence of these algorithms and Lemma \ref{lem_reduction}, Theorem \ref{thm_ub} follows.

\subsection{A Deterministic Algorithm for MPMD on Trees}\label{sec_MPMD_tree}

Suppose the tree metric is given by an edge-weighted tree $T$ rooted at an arbitrary vertex $r$. For a vertex $u$, let $T_u$ denote the maximal subtree of $T$ rooted at $u$, $e_u$ denote the edge between $u$ and its parent, and $d_u$ denote the weight of $e_u$ ($d_r$ is defined to be infinite). Let $h$ be the height of the tree, that is, the maximum of the number of vertices in the path between $r$ and any vertex $u$. We assume, without loss of generality, that the requests are given only at the leaves of $T$. (If not, we pretend as if each non-leaf vertex $u$ has a child $u'$ at a distance zero which is a leaf, and the requests are given at $u'$ instead of $u$.)

The algorithm maintains a forest $F\subseteq T$, and we say that an edge has been \textit{bought} if it is in $F$. Initially, $F$ is empty. As soon as there are two requests at vertices $u$ and $v$ such that the entire path between $u$ and $v$ is bought, we connect the two requests, and the edges on the path are removed from $F$. We say that a new \textit{phase} begins at vertex $u$ when the edge $e_u$ between $u$ and its parent is used to connect requests. Of course, the phases of the vertices need not be aligned. 
We say that a vertex $u$ is \textit{saturated} if the edge $u$ has been bought ($r$ is never saturated, by definition), else, we say that $u$ is \textit{unsaturated}. We say that a vertex $u$ is \textit{odd} if the number of pending requests in $T_u$ is odd, else we say that $u$ is \textit{even}. Each vertex $u$ (including $r$) has a counter $z_u$, initially zero, which increases at a unit rate if $u$ is unsaturated and odd; otherwise $z_u$ is frozen. For $u\neq r$, as soon as the value of $z_u$ becomes equal to an integral multiple of $2d_u$, the edge $e_u$ between $u$ and its parent is bought, i.e.\ included in $F$, $u$ becomes saturated, and $z_u$ is frozen. When this edge is eventually used, $u$ becomes unsaturated again.

For analysis, let $y_u$ denote the final value of the counter $z_u$ at the end of the input. We will separately relate the connection cost as well as the delay cost of the algorithm to $\sum_uy_u$, and then relate $\sum_uy_u$ to the cost of the adversary.

\begin{lemma}\label{lem1}
The connection cost of the algorithm is at most $\left(\sum_uy_u\right)/2$.
\end{lemma}

\begin{proof}
For an arbitrary vertex $u$, recall that $e_u$ is the edge between $u$ and its parent and $d_u$ is its weight. Between two consecutive usages of $e_u$ to connect requests, $z_u$ increases by exactly $2d_u$. Thus, the number of times $e_u$ is used to connect requests is $\lfloor y_u/(2d_u)\rfloor$. As a consequence, the connection cost of the algorithm is $\sum_ud_u\cdot\lfloor y_u/(2d_u)\rfloor\leq\left(\sum_uy_u\right)/2$.
\end{proof}

In order to bound the delay cost of the algorithm, we need to make the following observations.
\begin{observation}\label{obs_parity}
At any time, an odd non-leaf vertex has at least one odd child. If an even non-leaf vertex has an odd child, then it has another odd child.
\end{observation}
\begin{observation}\label{obs_component}
At any time, except for the time instants when requests are paired up, each connected component of $F$ has at most one vertex with a pending request.
\end{observation}

\begin{lemma}\label{lem2}
The delay cost of the algorithm is at most $2\sum_uy_u$.
\end{lemma}

\begin{proof}
At any time, except for the time instants when requests are paired up, let $L$ denote the set of leaves with a pending request, and $A$ denote the set of odd unsaturated vertices. We will define a function $f:L\longrightarrow A$ such that for any $a\in A$, $|f^{-1}(a)|\leq 2$. We can then charge the waiting time cost of any $l\in L$ to the increase in $z_{f(l)}$. Note that by definition $f(l)\in A$, and hence, $dz_{f(l)}/dt=1$. Furthermore, since $|f^{-1}(a)|\leq 2$, the waiting time of at most two requests is charged to the increase in $z_a$, for any $a$. This proves the lemma.

Here is how we construct the function $f$. Let $l\in L$. Consider the component $C_l$ of $F$ containing the vertex $l$. This is a subtree of $T$. Note that its root $r_l$ is unsaturated, otherwise the edge $e_{r_l}$ between $r_l$ and its parent would also be in $F$, and $r_l$ would not be the root of $C_l$. If $r_l$ is odd, then $r_l\in A$, and we define $f(l)=r_l$ in this case. Suppose now that $r_l$ is even. Trace the path $l=u_0,u_1,\ldots$ from $l$ upward towards $r_l$ until the first even vertex, say $u_n$ ($n>1$ since $l$ is odd). Since $u_n$ is an even vertex with an odd child $u_{n-1}$, by Observation \ref{obs_parity}, $u_n$ has another odd child $v_0\neq u_{n-1}$. We now define a path $p=(v_0,\ldots,v_j)$ of odd vertices such that $v_0,\ldots,v_{j-1}\in C_l$, $v_j\notin C_l$, and $v_j$ is unsaturated. We then set $f(l)=v_j$.

If $v_0$ is outside $C_l$, then $v_0$ is the last on $p$. If $v_0$ is inside $C_l$, then $v_0$ cannot be a leaf of $T$. (Otherwise, since $v_0$ is odd, there is a pending request at $v_0$, and then $v_0$ and $l$ belong to the same component $C_l$, which contradicts Observation \ref{obs_component}.) Thus, $v_0$ has an odd child. Call this child $v_1$ and add it to $p$. If $v_1\notin C_l$, then $v_1$ is the last on $p$. Otherwise, we continue extending $p$ in the same manner. Note that this cannot go on indefinitely, and $p$ must terminate. Since we terminate $p$ as soon as we step out of $C_l$, $v_j$ is not in $C_l$, but $v_{j-1}$, the parent of $v_j$ is in $C_l$. Thus, the edge $e_{v_j}$ between $v_j$ and its parent is not in $F$, which means that $v_j$ is unsaturated. Since $v_j$ is odd and unsaturated, $v_j\in A$, making the definition $f(l)=v_j$ legal.

We are left to prove that $f^{-1}(a)\leq2$ for all $a\in A$. Suppose for contradiction that $f^{-1}(a)>2$. Then $a$ has two pre-images, say $l$ and $l'$, which are either both inside $T_a$, or both outside $T_a$. In the former case, $a$ is the root of both $C_l$ as well as $C_{l'}$, which means $C_l=C_{l'}$, that is, $l$ and $l'$ belong to the same connected component of $F$, contradicting Observation \ref{obs_component}. In the latter case, the parent of $a$ is in both $C_l$ as well as $C_{l'}$, by the construction of $f$, again contradicting Observation \ref{obs_component}.
\end{proof}

We need to relate $\sum_uy_u$ to the cost of an arbitrary solution $\sol$ to the instance. For this, let $x_u$ be the total delay cost incurred by $\sol$ due to requests inside $T_u$, and $x'_u$ be the total connection cost incurred by $\sol$ for using the edge between $u$ and its parent.

\begin{lemma}\label{lem3}
For all vertices $u$, $y_u\leq2(x_u+x'_u)$.
\end{lemma}

\begin{proof}
Call $u$ \textit{misaligned} if the parity of the number of algorithm's pending requests inside $T_u$ and the parity of the number of $\sol$'s pending requests inside $T_u$ do not agree; otherwise call $u$ \textit{aligned}. As long as $u$ is aligned, whenever $z_u$ is increasing, the adversary has a pending request inside $T_u$, which means $x_u$ is also increasing at least at a unit rate. We ignore the increase in $x_u$ when $u$ is misaligned. The alignment status flips only when the edge between $u$ and its parent is used either by the algorithm or by the adversary.

Say that event $E$ occurs when the adversary pairs up a request inside $T_u$ to a request outside $T_u$. Suppose $E$ occurs $k$ times. Then $x'_u=kd_u$. For every phase in which $E$ occurs, imagine that the occurrences of $E$ are shifted to one of the boundaries of that phase, so that $u$ was misaligned for the entire phase. Ignore the contribution of the delay to $x_u$ in these phases. This can only decrease $x_u$. As a result of this, for every phase, $u$ is either aligned or misaligned in the entire phase. Also, in case the last phase is incomplete, if $u$ was aligned, ignore its contribution to $x_u$ and $y_u$, else pretend as if it was complete. This can only increase the ratio of $y_u$ to $x_u+x_u'$.

Let $n_i$ (resp.\ $n_0$) be the number of phases between the $i^{\text{\tiny{th}}}$ and the $(i+1)^{\text{\tiny{th}}}$ occurrence of $E$ (resp.\ before the first occurrence of $E$). Then at least $\lfloor n_i/2\rfloor$ (resp.\ $\lceil n_0/2\rceil$) of these must be phases in which $u$ was aligned, since the alignment status flips exactly when a new phase begins (resp.\ and since $u$ is aligned in the first phase if $n_0>0$). The contribution of all these phases to $x_u$ will be at least $2d_u\left(\lceil n_0/2\rceil+\sum_{i=1}^{k}\lfloor n_i/2\rfloor\right)\geq2d_u\left(\sum_{i=0}^kn_i/2-k/2\right)=d_u\cdot\left(\sum_{i=0}^kn_i-k\right)$. Adding $x'_u$, we have $x_u+x'_u\geq d_u\left(\sum_{i=1}^kn_i-k\right)+kd_u=d_u\sum_{i=1}^kn_i=y_u/2$,
where the last equality holds since $z_u$ increases by $2d_u$ in every phase. Thus, we have $y_u\leq2(x_u+x'_u)$.
\end{proof}

Finally, we relate $\sum_u(x_u+x'_u)$ to the cost of the solution $\sol$. Denoting the distance function of the tree metric by $d$, recall that $\sol_d$ and $\sol_t$ denote the connection cost and the delay cost of $\sol$ respectively.

\begin{lemma}\label{lem4}
$\sum_u(x_u+x'_u)\leq\sol_d+h\cdot\sol_t$.
\end{lemma}

\begin{proof}
$\sum_ux'_u$ is clearly equal to the connection cost $\sol_d$ of the solution. A pending request in the solution $\sol$ at a leaf $l$ contributes to the increase in $x_u$ if an only if $u$ is an ancestor of $l$. Thus, each pending request in $\sol$ contributes at most $h$ to the rate of increase of $\sum_ux_u$, and one to the rate of increase of $\sol_t$. Therefore, $\sum_ux_u$ is at most $h$ times the delay cost $\sol_t$ of the solution.
\end{proof}

The competitiveness of the algorithm now follows easily.

\begin{theorem}\label{thm_tree}
The algorithm for MPMD on tree metrics is $(5,5h)$-competitive, and hence, $5h$-competitive.
\end{theorem}

\begin{proof}
From Lemmas \ref{lem1} and \ref{lem2}, the algorithm's total cost is at most $\frac{5}{2}\sum_uy_u$. By Lemma \ref{lem3}, this is at most $5\sum(x_u+x'_u)$, which by Lemma \ref{lem4}, is at most $5\sol_d+5h\cdot\sol_t$. Therefore, the algorithm is $(5,5h)$-competitive.
\end{proof}

Finally, we prove Theorem \ref{thm_ub} for MPMD using the above theorem and the reduction result (Lemma \ref{lem_reduction}).

\begin{proof}[Proof of Theorem \ref{thm_ub} for MPMD]
Given an arbitrary $n$-point metric space $\mathcal{M}$, Lemma \ref{lem_embed} ensures that $\mathcal{M}$ can be embedded into a distribution $\mathcal{D}$ supported on metrics induced by trees of height $O(\log n)$. Theorem \ref{thm_tree} ensures that there is an $(O(1),O(\log n))$-competitive algorithm for MPMD on every tree metric in the support of $\mathcal{D}$, i.e.\ the algorithm always returns a solution with cost at most $O(1)$ times the connection cost plus $O(\log n)$ times the delay cost of any solution. Therefore, by Lemma \ref{lem_reduction} there is an $(O(\log n),O(\log n))$-competitive (equivalently, $O(\log n)$-competitive) algorithm for MPMD on $\mathcal{M}$. This algorithm samples a tree metric $\mathcal{M}'$ from $\mathcal{D}$, and runs the deterministic algorithm for MPMD on tree metrics on $\mathcal{M}'$ to process the online input.
\end{proof}

\subsection{A Deterministic Algorithm for MBPMD on Trees}

Suppose the tree metric is given by an edge-weighted tree $T$ rooted at an arbitrary vertex $r$. As before, for a vertex $u$, let $T_u$ denote the maximal subtree of $T$ rooted at $u$, $e_u$ denote the edge between $u$ and its parent, and $d_u$ denote the weight of $e_u$ ($d_r$ is defined to be infinite). Let $h$ be the height of the tree, that is, the maximum of the number of vertices in the path between $r$ and any vertex $u$. Again we assume, without loss of generality, that the requests are given only at the leaves of $T$. Let $\lca(u,v)$ denote the lowest common ancestor of vertices $u$ and $v$ in the tree. Define the \textit{surplus} of a vertex $v$ to be the number of positive requests minus the number of negative requests in $T_v$, and denote it by $\sur(v)$. (Note that $\sur(v)$ can be negative.)

The algorithm maintains two forests $F^+,F^-\subseteq T$. Initially, both $F^+$ and $F^-$ are empty. We say that a vertex is \textit{positively saturated} (resp.\ \textit{negatively saturated}) if the edge between it and its parent is in $F^+$ (resp.\ $F^-$). ($r$ is never saturated, by definition), else, we say it is \textit{positively unsaturated} (resp.\ \textit{negatively unsaturated}).\footnote{Note $F^+$ and $F^-$ are not necessarily disjoint, and therefore, a vertex can be both positively as well as negatively saturated at the same time.} Each vertex $u$ (including $r$) has two counters $z^+_u$ and $z^-_u$, initially zero. Counter $z^+_u$ (resp.\ $z^-_u$) increases at the rate $\sur(u)$ (resp.\ $-\sur(u)$) if $u$ is positively unsaturated (resp.\ negatively unsatuared) and $\sur(u)>0$ (resp.\ $\sur(u)<0$); otherwise $z^+_u$ (resp.\ $z^-_u$) is frozen. For $u\neq r$, as soon as the value of $z^+_u$ (resp.\ $z^-_u$) becomes equal to $2d_u$, the edge $e_u$ between $u$ and its parent is added to $F^+$ (resp.\ $F^-$), $u$ becomes positively saturated (resp.\ negatively saturated), and $z^+_u$  (resp.\ $z^-_u$) is frozen.

As soon as there is a positive request at a vertex $u^+$ and a negative request at a vertex $u^-$ such that the entire path between $u^+$ and $\lca(u^+,u^-)$ is contained in $F^+$, and the entire path between $u^-$ and $\lca(u^+,u^-)$ is contained in $F^-$, we connect the two requests, remove edges on the path from $u^+$ to $u^-$ from both $F^+$ as well as $F^-$, and reset all counters associated with these edges to $0$. We say that a new \textit{phase} begins at vertex $u$ when the edge $e_u$ between $u$ and its parent is used to connect requests. 


For analysis, imagine a variable $Z^+_u$ (resp.\ $Z^-_u$) for every $u$, which increases at the same rate as $z^+_u$ (resp.\ $z^-_u$) during the run of the algorithm, but which is never reset to zero. Let $y^+_u$ (resp.\ $y^-_u$) denote the final value of $Z^+_u$ (resp.\ $Z^-_u$). We will separately relate the connection cost as well as the delay cost of the algorithm to $\sum_u(y^+_u+y^-_u)$, and then relate $\sum_u(y^+_u+y^-_u)$ to the cost of the adversary.

\begin{lemma}\label{lem1b}
The connection cost of the algorithm is at most $\frac{1}{2}\sum_u(y^+_u+y^-_u)$.
\end{lemma}

\begin{proof}
For an arbitrary vertex $u$, recall that $e_u$ is the edge between $u$ and its parent and $d_u$ is its weight. Between two consecutive usages of $e_u$ to connect requests, either $Z^+_u$ or $Z^-_u$ increases by exactly $2d_u$. This implies the claim.
\end{proof}

In order to bound the delay cost of the algorithm, we need to bound the number of pending requests at any moment by the rate of the increase of the counters. We do this by induction on the tree $T$. For this, we need the following definition.

\begin{definition}
A \textit{snapshot} $\mathcal{S}$ is a tuple $(T,R,F^+,F^-)$, where
\begin{itemize}
\item $T$ is a rooted tree.
\item $R$ is a function from the leaves of $T$ to $\mathbb{Z}$, where $R(l)$ denotes the signed number of requests at leaf $l$.
\item $F^+,F^-\subseteq T$ are forests.
\end{itemize}
Let $\rho^+(\mathcal{S})=\sum_{l\text{: leaf of }T}\max(R(l),0)$ (resp.\ $\rho^-(\mathcal{S})=\sum_{l\text{: leaf of }T}\max(-R(l),0)$) denote the number of positive (resp.\ negative) requests in the snapshot $\mathcal{S}$, and define 
\[\zeta(\mathcal{S})=\sum_{u:e_u\notin F^+}\max(\sur(u),0)+\sum_{u:e_u\notin F^-}\max(-\sur(u),0)\]
Call the snapshot \textit{valid} if there is no pair of leaves $l^+$ and $l^-$ such that the following are satisfied.
\begin{itemize}
\item $R(l^+)>0$ and $R(l^-)<0$.
\item The path between $l^+$ and $\lca(l^+,l^-)$ is contained in $F^+$, and the path between $l^-$ and $\lca(l^+,l^-)$ is contained in $F^-$.
\end{itemize}
\end{definition}

Observe that if our algorithm has snapshot $\mathcal{\mathcal{S}}$, then the total rate of increase of the counters would be $\zeta(\mathcal{S})$. Our goal is to bound the number of pending requests by the rate of increase of the counters. Therefore, we bound $\rho^+(\mathcal{S})$ and $\rho^-(\mathcal{S})$ by $\zeta(\mathcal{S})$. Note that the algorithm is defined in such a way that as soon as its snapshot becomes invalid, requests get eliminated and the snapshot becomes valid again.

\begin{lemma}\label{lem_f_rec}
Given a valid snapshot $\mathcal{S}=(T,R,F^+,F^-)$, we have $\rho^+(\mathcal{S})\leq\zeta(\mathcal{S})$ and $\rho^-(\mathcal{S})\leq\zeta(\mathcal{S})$.
\end{lemma}

\begin{proof}
We only prove the upper bound on $\rho^+(\mathcal{S})$. The upper bound on $\rho^-(\mathcal{S})$ follows by a symmetric argument. The proof is by induction on the structure of $T$. When $T$ has a single vertex, the claim is obvious. Suppose $T$ has more than one vertices. Let $v$ be an arbitrary vertex of $T$ such that all children of $v$ are leaves (such a vertex always exists). Let $R^+=\sum_{l\text{: child of }v}\max(R(l),0)$ and $R^-=\sum_{l\text{: child of }v}\max(-R(l),0)$ be the number of positive and negative requests respectively in the subtree rooted at $v$. Then $\sur(v)=R^+-R^-$.

We split the proof into several cases, depending on the sign of $\sur(v)$. In each case, we construct another snapshot $\mathcal{S'}=(T',R',F'^+,F'^-)$ as follows. $T'$ is $T$ with the children of $v$ removed, so that $v$ is a leaf of $T'$. $R'$ is the same as $R$, except that $R'(v)$ is equal to $\sur(v)$ in $\mathcal{S}$. This ensures that the surplus of any vertex (except for the children of $v$) is the same in $\mathcal{S}'$ and $\mathcal{S}$. $F'^+$ and $F'^-$ are restrictions of $F^+$ and $F^-$ respectively to $T'$. In cases where the validity of such a snapshot is not ensured, we make minor adjustments to $F^+$ or $F^-$ and restore validity. Using the induction hypothesis for the snapshot $\mathcal{S}'$, we prove the claim.

Since $\mathcal{S}$ is a valid snapshot to begin with, either none of the children of $v$ with positive requests is connected to $v$ in $F^+$, or none of the children of $v$ with negative requests is connected to $v$ in $F^-$. Thus, either all the leaves with positive requests, or all the leaves with negative requests, (or all the leaves) under $v$ contribute to $\zeta(\mathcal{S})$. However, in any case, none of the leaves under $v$ contributes to $\zeta(\mathcal{S}')$.

\noindent \textbf{Case 1: $\sur(v)=0$.} Then $R^+=R^-$. Observe that snapshot $\mathcal{S'}$ is valid in this case. We have,
\[\rho^+(\mathcal{S})=\rho^+(\mathcal{S}')+R^+\]
As observed earlier, either all the positive or all the negative requests under $v$ contribute to $\zeta(\mathcal{S})$, but none of them contributes to $\zeta(\mathcal{S}')$. Therefore,
\[\zeta(\mathcal{S})\geq\zeta(\mathcal{S}')+R^+\]
By induction, $\rho^+(\mathcal{S'})\leq\zeta(\mathcal{S'})$. Putting everything together, we get that $\rho^+(\mathcal{S})\leq\zeta(\mathcal{S})$.

\noindent \textbf{Case 2a: $\sur(v)>0$ and one of the children of $v$ with positive requests is connected to $v$ in $F^+$.} Again, observe that snapshot $\mathcal{S'}$ is valid in this case. We have,
\[\rho^+(\mathcal{S})=\rho^+(\mathcal{S}')-\sur(v)+R^+=\rho^+(\mathcal{S}')+R^-\]
Since one of the children of $v$ with positive requests is connected to $v$ in $F^+$, none of the children of $v$ with negative requests is connected to $v$ in $F^-$. These negative requests contribute $R^-$ to $\zeta(\mathcal{S})$, but not to $\zeta(\mathcal{S}')$. By construction, other than the children of $v$, the contribution of every vertex to $\zeta(\mathcal{S})$ and $\zeta(\mathcal{S}')$ is the same. Therefore,
\[\zeta(\mathcal{S})\geq\zeta(\mathcal{S}')+R^-\]
By induction, $\rho^+(\mathcal{S'})\leq\zeta(\mathcal{S'})$. Putting everything together, we get that $\rho^+(\mathcal{S})\leq\zeta(\mathcal{S})$.

\noindent \textbf{Case 2b: $\sur(v)>0$ and none of the children of $v$ with positive requests is connected to $v$ in $F^+$.} In this case, $\mathcal{S}'$ need not be a valid snapshot, but observe that this happens only when $v$ is not the root, and the edge between $v$ and its parent belongs to $F'^+$. Remove that edge from $F'^+$, and observe that this makes $\mathcal{S}'$ valid. The side-effect of this tweak is that $v$ now contributes $\sur(v)$ to $\zeta(\mathcal{S}')$, but does not contribute anything to $\zeta(\mathcal{S})$. Recall that the children of $v$ with positive requests contribute to $\zeta(\mathcal{S})$ but not $\zeta(\mathcal{S}')$, and observe that all vertices, other than $v$ and its children, contribute equally to $\zeta(\mathcal{S})$ and $\zeta(\mathcal{S}')$. Thus,
\[\zeta(\mathcal{S})\geq\zeta(\mathcal{S}')-\sur(v)+R^+\]
As in case 2a, we have
\[\rho^+(\mathcal{S})=\rho^+(\mathcal{S}')-\sur(v)+R^+\]
Again, by induction, $\rho^+(\mathcal{S'})\leq\zeta(\mathcal{S'})$. Putting everything together, we get that $\rho^+(\mathcal{S})\leq\zeta(\mathcal{S})$.

\noindent \textbf{Case 3a: $\sur(v)<0$ and one of the children of $v$ with negative requests is connected to $v$ in $F^-$.} Again, observe that snapshot $\mathcal{S'}$ is valid in this case. Noting that $v$ has $|\sur(v)|$ negative requests in $\mathcal{S}'$, we have,
\[\rho^+(\mathcal{S})=\rho^+(\mathcal{S}')+R^+\]
Since one of the children of $v$ with negative requests is connected to $v$ in $F^-$, none of the children of $v$ with positive requests is connected to $v$ in $F^+$. These positive requests contribute $R^+$ to $\zeta(\mathcal{S})$, but not to $\zeta(\mathcal{S}')$. By construction, other than the children of $v$, the contribution of every vertex to $\zeta(\mathcal{S})$ and $\zeta(\mathcal{S}')$ is the same. Therefore,
\[\zeta(\mathcal{S})\geq\zeta(\mathcal{S}')+R^+\]
By induction, $\rho^+(\mathcal{S'})\leq\zeta(\mathcal{S'})$. Putting everything together, we get that $\rho^+(\mathcal{S})\leq\zeta(\mathcal{S})$.

\noindent \textbf{Case 3b: $\sur(v)<0$ and none of the children of $v$ with negative requests is connected to $v$ in $F^-$.} In this case, $\mathcal{S}'$ need not be a valid snapshot, but observe that this happens only when $v$ is not the root, and the edge between $v$ and its parent belongs to $F'^-$. Remove that edge from $F'^-$, and observe that this makes $\mathcal{S}'$ valid. As in Case 2b, a side-effect of this tweak is that $v$ now contributes $-\sur(v)>0$ to $\zeta(\mathcal{S}')$, but does not contribute anything to $\zeta(\mathcal{S})$. Recall that the children of $v$ with negative requests contribute to $\zeta(\mathcal{S})$ but not $\zeta(\mathcal{S}')$, and observe that all vertices, other than $v$ and its children, contribute equally to $\zeta(\mathcal{S})$ and $\zeta(\mathcal{S}')$. Thus,
\[\zeta(\mathcal{S})\geq\zeta(\mathcal{S}')-(-\sur(v))+R^-=\zeta(\mathcal{S}')+\sur(v)+R^-=\zeta(\mathcal{S}')+R^+\]
As in case 3a, we have
\[\rho^+(\mathcal{S})=\rho^+(\mathcal{S}')+R^+\]
Again, by induction, $\rho^+(\mathcal{S'})\leq\zeta(\mathcal{S'})$. Putting everything together, we get that $\rho^+(\mathcal{S})\leq\zeta(\mathcal{S})$.
\end{proof}

As an immediate consequence of the above lemma, we have the following bound on the delay cost of our algorithm.

\begin{lemma}\label{lem2b}
The delay cost of the algorithm is at most $2\sum_u(y^+_u+y^-_u)$.
\end{lemma}

\begin{proof}
Consider a moment and let $\mathcal{S}$ be the algorithm's snapshot at this moment. Then the delay cost increases at the rate $\rho^+(\mathcal{S})+\rho^-(\mathcal{S})$, the number of pending requests. The counter $z^+_u$ (resp.\ $z^-_u$) increases at the rate $\max(\sur(u),0)$ (resp.\ $\max(-\sur(u),0)$) if and only if $u$ is the root of a component in $F^+$ (resp.\ $F^-$). Thus, the rate of the increase of $\sum_u(Z^+_u+Z^-_u)$ is equal to $\zeta(\mathcal{S})$. Thus, the claim follows, since $y^+_u$ and $y^-_u$ are the final values of $Z^+_u$ and $Z^-_u$ respectively.
\end{proof}

We need to relate $\sum_u(y^+_u+y^-_u)$ to the cost of an arbitrary solution $\sol$ to the instance. For this, let $x_u$ be the total delay cost incurred by $\sol$ due to requests inside $T_u$, and $x'_u$ be the total cost incurred by $\sol$ for using the edge between $u$ and its parent.

\begin{lemma}\label{lem3b}
For all vertices $u$, $y^+_u+y^-_u\leq4(x_u+x'_u)$.
\end{lemma}

\begin{proof}
We use the technique of potentials functions. We design a potential function $\phi$ such that in each phase, the changes $\Delta(Z^+_u+Z^-_u)$, $\Delta\phi$, and $\Delta(x_u+x'_u)$ satisfy
\begin{equation}\label{eqn_pot}
\Delta(y^+_u+y^-_u)+\Delta\phi\leq4\Delta(x_u+x'_u)
\end{equation}
and $\phi=0$ in the beginning as well as at the end. Summing over all phases, we get the result. 

At any point of time, let $\sur'(u)$ denote the surplus of vertex $u$ resulting from the adversary's solution. Define $\phi=4d_u\cdot|\sur'(u)-\sur(u)|$. Clearly, in the beginning as well as in the end, we have $\sur(u)=\sur'(u)=0$, and thus, $\phi=0$. Observe that $\sur'(u)-\sur(u)$ (and hence, $\phi$) remains unchanged when new requests are given. The only events resulting in a change in $\sur'(u)-\sur(u)$ are either $\sol$ or the algorithm connecting a request in $T_u$ to one outside $T_u$. Also, then $x_u$ increases at a rate at least $|\sur'(u)|$.

In each phase of vertex $u$, each of $Z^+_u$ and $Z^-_u$ increases by at most $2d_u$, and therefore, $\Delta(Z^+_u+Z^-_u)\leq4d_u$. Except possibly the last phase, in every phase, at least one of $Z^+_u$ and $Z^-_u$ increases by exactly $2d_u$, and the phase ends with the algorithm connecting a request inside $T_u$ to one outside. We call such a phase \textit{complete}. If we have $z^+_u+z^-_u>0$ at the end of the algorithm, we call the last phase \textit{incomplete}. We first prove that (\ref{eqn_pot}) holds for complete phases, and then for incomplete phases. Let $k\geq0$ denote the (absolute) number of requests in $T_u$ which $\sol$ connected to requests outside $T_u$ during the phase. Thus, $\Delta x'_u\geq kd_u$.


Consider any complete phase of vertex $u$ and, without loss of generality, assume that the phase ends due to a positive request in $T_u$ getting connected to a negative request outside $T_u$
. This means that $z^+_u$ increases from $0$ to $2d_u$ in the phase. Since the only events resulting in a change in $\sur'(u)-\sur(u)$ are either $\sol$ or the algorithm connecting a request in $T_u$ to one outside, we have
\begin{equation}\label{eqn_triangle}
\Delta|\sur'(u)-\sur(u)|\leq|\Delta(\sur'(u)-\sur(u))|\leq k+1
\end{equation}

First, consider the case where $\Delta|\sur'(u)-\sur(u)|=k+1$, and therefore, $\Delta\phi=4(k+1)\cdot d_u$ in the phase. Now both inequalities in (\ref{eqn_triangle}) are tight. 
Because the second inequality is tight, 
all the $k$ requests inside $T_u$ which $\sol$ connected outside must be negative, and thus, $\Delta(\sur'(u)-\sur(u))=k+1>0$. Also, $\sur'(u)-\sur(u)$ never decreases during the phase. Because the first inequality in (\ref{eqn_triangle}) is tight, the sign of $\sur'(u)-\sur(u)$ at the beginning of the phase must be same as that of $\Delta(\sur'(u)-\sur(u))$, implying $\sur'(u)-\sur(u)\geq0$ initially. Since $\sur'(u)-\sur(u)$ never decreases, we have $\sur'(u)-\sur(u)\geq0$ throughout the phase. Therefore, at any moment when $z^+_u$ was increasing, we have $\sur'(u)\geq\sur(u)>0$.
Thus, the rate of increase of $x_u$ is always at least as much as the rate of increase of $z^+_u$. Since $z^+_u$ increases by $2d_u$, we have $\Delta x_u\geq2d_u$. Therefore,
\[\Delta(Z^+_u+Z^-_u)+\Delta\phi\leq4d_u+4(k+1)\cdot d_u=4(2d_u+kd_u)\leq4\Delta(x_u+x'_u)\]

Next, suppose that $\Delta|\sur'(u)-\sur(u)|<k+1$. Observe that the parity of $\sur'(u)-\sur(u)$ changes $k+1$ times during the phase: each time when the algorithm or $\sol$ connects a request in $T_u$ to one outside. Thus, if $\Delta|\sur'(u)-\sur(u)|$ is not $k+1$, it must be at most $k-1$, which means $\Delta\phi\leq4(k-1)\cdot d_u$. Therefore,
\[\Delta(Z^+_u+Z^-_u)+\Delta\phi\leq4d_u+4(k-1)\cdot d_u=4kd_u=4\Delta x'_u\leq4\Delta(x_u+x'_u)\]
Thus, in any case, (\ref{eqn_pot}) holds for any complete phase.

Finally consider the last incomplete phase. Note that at the end of the algorithm $\sur(u)=\sur'(u)=0$, and hence, $\phi=0$. Since $\phi$ is non-negative by definition, we have $\Delta\phi\leq0$. If $k>0$, then $\Delta(x_u+x'_u)\geq\Delta x'_u\geq kd_u\geq d_u$. Since $\Delta(Z^+_u+Z^-_u)\leq4d_u$, (\ref{eqn_pot}) holds. On the other hand, if $k=0$, then $\sur'(u)-\sur(u)$ stays constant in the phase. Since it is zero finally, it is zero throughout the phase. Thus, $\sur'(u)=\sur(u)$ in the entire phase. This means $\Delta(Z^+_u+Z^-_u)=\Delta x_u$, again implying (\ref{eqn_pot}).
\end{proof}

As before, we relate $\sum_u(x_u+x'_u)$ to the cost of the solution $\sol$. Denoting the distance function of the tree metric by $d$, recall that $\sol_d$ and $\sol_t$ denote the connection cost and the delay cost of $\sol$ respectively.

\begin{lemma}\label{lem4b}
$\sum_u(x_u+x'_u)\leq\sol_d+h\cdot\sol_t$.
\end{lemma}

\begin{proof}
Same as the proof of Lemma \ref{lem4}.
\end{proof}

The competitiveness of the algorithm follows in an analogously as in the proof of Theorem \ref{thm_tree}.

\begin{theorem}\label{thm_b_tree}
The algorithm for MBPMD on tree metrics is $(10,10h)$-competitive, and hence, $10h$-competitive.
\end{theorem}

\begin{proof}
From Lemmas \ref{lem1b} and \ref{lem2b}, the algorithm's total cost is at most $\frac{5}{2}\sum_u(y^+_u+y^-_u)$. By Lemma \ref{lem3b}, this is at most $10\sum(x_u+x'_u)$, which by Lemma \ref{lem4b}, is at most $10\sol_d+10h\cdot\sol_t$. Therefore, the algorithm is $(10,10h)$-competitive.
\end{proof}


\begin{proof}[Proof of Theorem \ref{thm_ub} for MBPMD]
Same as the proof of Theorem \ref{thm_ub} for MPMD presented at the end of Section \ref{sec_MPMD_tree} (using Theorem \ref{thm_b_tree} instead of Theorem \ref{thm_tree}).
\end{proof}

\section{The Lower Bounds}\label{sec_lb}

The focus of this section is to prove the following lower bound results.

\begin{theorem}\label{thm_lb}
There is an $n$-point metric space on which any randomized algorithm for MPMD (resp.\ MBPMD) has competitive ratio $\Omega(\sqrt{\log n})$ (resp.\ $\Omega(\log^{1/3}n)$), against an oblivious adversary.
\end{theorem}

The required metric space is given by $n$ equally spaced points on the real interval $[0,1]$, where $n$ is even. All asymptotic notation in this section is with respect to $n\rightarrow\infty$. Note that the metric space of $n$ equally spaced points is trivially a tree metric given by a tree of height $n/2$. We give a distribution on input instances of MPMD (resp.\ MPMBD) on which the expected cost incurred by any deterministic online algorithm is $\Omega(\sqrt{\log n})$ (resp.\ $\Omega(\log^{1/3}n)$) times the cost of the optimum solution. The construction of the distribution is in several phases, and we need the following key lemma to analyze each phase.

\begin{lemma}\label{lem_onephase}
Suppose $A\subseteq[0,1]$ is an arbitrary finite set of requests, and $B\subseteq[0,1]$ is a finite set of requests spaced at least a distance $d$ apart. Suppose $C\subseteq A\cup B$ is such that $|(A\cup B)\setminus C|$ is even. Then the cost of the optimum perfect matching on $(A\cup B)\setminus C$ is at least $d\times (|B|-(|A|+|C|))/2$.
\end{lemma}

\begin{proof}
The set $(A\cup B)\setminus C$ contains at least $|B|-|C|$ requests from $B$. Out of these requests, at most $|A|$ requests can be matched with requests in $A$. Therefore, at least $|B|-(|A|+|C|)$ requests are paired up among themselves, resulting in at least $(|B|-(|A|+|C|))/2$ pairs of requests, all from $B$. The distance between every pair of requests in $B$ is at least $d$. Therefore, the claim follows.
\end{proof}

When we use the above lemma, we will actually ensure that $|B|\geq2|A|$, $|C|/|B|=o(1)$, and $d\approx1/|B|$. So the cost of the matching is at least $1/4-o(1)$.

\subsection{The $\Omega(\sqrt{\log n})$ Lower bound for MPMD}\label{subsec_lb}

We give a distribution on input instances of MPMD which ensures that any deterministic algorithm pays $\Omega(\sqrt{\log n})$ in expectation, while the instances have solutions of cost $O(1)$. 

The construction of the bad distribution on inputs depends on three parameters which, in turn, depend on $n$: the number of phases, denoted by $r+1$, a ``decay factor'' $\rho$, which bounds the ratio of the number of new requests given in any phase to that given in the following phase, and $a$, which bounds the cost of serving the requests in each phase, in the optimum solution. We will choose the values of the parameters such that $r=\omega(1)$, $\rho=\omega(r)$, and $\rho^{2r}=o(n)$.
The distribution $\mathcal{D}$ on input sequences is generated as follows. 
\begin{enumerate}
\item Initially, $n_0:=n$, where $n$ is even, and $S_0:=\{1/n,2/n,\ldots,1\}$.
\item For $i=0,\ldots,r$,
\begin{enumerate}
\item Sample $y_i$ from $U[0,1]$, the uniform distribution on the interval $[0,1]$.
\item $t_i:=\frac{a\rho^{1+y_i}}{n_i}$, $n_{i+1}:=2\left\lfloor\frac{a}{t_i}\right\rfloor=2\left\lfloor\frac{n_i}{\rho^{1+y_i}}\right\rfloor$.
\end{enumerate}
\item For $i=0,\ldots,r$,
\begin{enumerate}
\item Give requests at points in $S_i$. 
\item Construct $S_{i+1}\subseteq S_i$ by scanning $S_i$ in ascending order, and including every $\left\lfloor\frac{n_i}{n_{i+1}}\right\rfloor^{\text{\tiny{th}}}$ point.
\item Wait for time $t_i$ (and then move on to the next phase, if $i<r$). 
\end{enumerate}
\end{enumerate}

Let $d_i$ be the distance between consecutive requests of phase $i$. Then $d_0=1/n$ and $d_i\leq1/n_i$. From the construction of the random instance, the following observation is evident.

\begin{observation}\label{obs_support}
For all $y_0,\ldots,y_i$, $t_i\in\left[\frac{a\rho}{n_i},\frac{a\rho^2}{n_i}\right]$ and $n_{i+1}\in\left[2\left\lfloor\frac{n_i}{\rho^2}\right\rfloor,2\left\lfloor\frac{n_i}{\rho}\right\rfloor\right]$.
\end{observation}

Since $\rho^{2r}=o(n)$, the above observation implies that we have enough supply of points at the beginning to support $r+1$ phases.

First, let us prove a bound on the cost of the optimum solution of the random instances constructed as above.

\begin{lemma}\label{lem_adv}
For any $y_0,\ldots,y_r$, the MPMD instance generated as above has a solution of cost at most $2ar+1$. 
\end{lemma}

\begin{proof}
Construct a solution as follows. For $i$ decreasing from $r$ to $1$, connect each unpaired request of phase $i$ to the request of phase $i-1$ located at the same point. This is possible because $S_{i}\subseteq S_{i-1}$. The connection cost of these pairs is zero. The cost paid for the delay is at most $n_it_{i-1}=2\left\lfloor\frac{a}{t_{i-1}}\right\rfloor\cdot t_{i-1}\leq2a$, for every $i$, and hence, at most $2ar$ after we sum over the phases. Finally, scan the set of unpaired requests of phase $0$ from left to right, and pair them up greedily. This results in at most a unit connection cost. Thus, the total cost of this solution is at most $2ar+1$.
\end{proof}

Let us now turn our attention to bounding from below the expected cost of an arbitrary deterministic online algorithm for MPMD on a random instance from the distribution $\mathcal{D}$. 

From the construction, it is intuitive to think that the distance $d_i$ between consecutive points in every $S_i$ is almost $1/n_i$. This is clearly true if $n_i$ is divisible by $n_{i+1}$ for all $i$. In the next lemma, we assert that this indeed holds, in spite of the accumulation of errors due to the repeated rounding down of $n_i/n_{i+1}$ in each phase, but we defer the proof to Appendix \ref{app_b}.

\begin{lemma}\label{lem_error}
For each $i$, $n_id_i=1-o(1)$.
\end{lemma}

Let $m_i$ denote the number of pending requests of the algorithm at the beginning of phase $i$ ($m_0=0$). Consider an imaginary scenario where, instead of giving the next batch of requests after time delay $t_i$, the adversary refrains from giving any further requests. In this scenario, let $x_i(t)$ denote the number of pending requests of the algorithm at time $t$ after the beginning of phase $i$. We call this the \textit{characteristic function} of the algorithm in phase $i$. Coming back to reality, since the algorithm is deterministic and online, the number of pending requests at time $t\leq t_i$ in phase $i$ is precisely $x_i(t)$. Thus, the number of requests carried over to the next phase is $m_{i+1}=x_i(t_i)$. Note that $m_i$ and $n_i$ are random variables and $x_i(\cdot)$ is a random function, all completely determined by $y_0,\ldots,y_{i-1}$. Let $A_i$ and $B_i$ denote the delay cost and the connection cost, respectively, paid by the algorithm in phase $i$ (including the delay and connection costs due to requests from previous phases, provided they lived long enough to see phase $i$). Then $A_i$ and $B_i$ are completely determined by $y_0,\ldots,y_i$.
Call the phase $i$ \textit{well-started} if $2m_i\leq n_i$. In the next two lemmas we consider the case where $\Pr_{y_i}\left[x_i(t_i)<\frac{a}{t_i}\right]$ is smaller and larger than $1/4$ respectively, and state the consequences.

\begin{lemma}\label{lem_bad}
For any $i$, and any $y_0,\ldots,y_{i-1}$, if $x_i(\cdot)$ is such that $\Pr_{y_i}\left[x_i(t_i)<\frac{a}{t_i}\right]<1/4$, then $\mathbb{E}_{y_i}[A_i]\geq(a\ln\rho)/8$.
\end{lemma}

Informally, the above lemma states that if $x(t)$ is at least $a/t$ for most $t$, then the expected value of the algorithm's delay cost is large. We defer the proof to Appendix \ref{app_b}.

\begin{lemma}\label{lem_good}
For any $i$, and any $y_0,\ldots,y_{i-1}$, if $x_i(\cdot)$ is such that $\Pr_{y_i}\left[x_i(t_i)<\frac{a}{t_i}\right]\geq1/4$, then
\begin{enumerate}
\item $\Pr_{y_i}[\text{Phase }i+1\text{ is well-started}]\geq1/4$.
\item Additionally, if phase $i$ is well-started, then $\mathbb{E}_{y_i}[B_i]\geq1/16-o(1)$. 
\end{enumerate} 
\end{lemma}

\begin{proof}
Suppose the event $x_i(t_i)<\frac{a}{t_i}$ happens. Then $x_i(t_i)\leq\left\lfloor\frac{a}{t_i}\right\rfloor$, because $x_i(t_i)$ is an integer. Since $m_{i+1}=x_i(t_i)$ and $n_{i+1}=2\left\lfloor\frac{a}{t_i}\right\rfloor$, we have $2m_{i+1}\leq n_{i+1}$, implying that the next phase is well-started.

Additionally, suppose the current phase is well-started, that is, $2m_i\leq n_i$. The number of requests remaining at the end of the phase is at most $\left\lfloor\frac{a}{t_i}\right\rfloor$, which is at most $n_i/\rho$ because $t_i\geq a\rho/n_i$ by Observation \ref{obs_support}. By Lemmas \ref{lem_onephase} and \ref{lem_error}, the algorithm must pay a connection cost of at least 
\[d_i\times\frac{n_i-(m_i+n_i/\rho)}{2}=n_id_i\times\frac{1-m_i/n_i-1/\rho}{2}\geq(1-o(1))\times\left(\frac{1}{4}-\frac{1}{2\rho}\right)=\frac{1}{4}-o(1)\]
since $\rho$ is $\omega(1)$. Thus, $\mathbb{E}_{y_i}[B_i]\geq(1/4-o(1))\cdot\Pr_{y_i}\left[x_i(t_i)<\left\lfloor\frac{a}{t_i}\right\rfloor\right]=1/16-o(1)$.
\end{proof}

\begin{observation}\label{obs_wellstarted}
By Lemma \ref{lem_bad} and part 2 of Lemma \ref{lem_good}, given that a phase is well-started, the algorithm must pay at least $\min((a\ln\rho)/8,1/16-o(1))$ in the phase, in expectation (regardless of whether $\Pr_{y_i}\left[x_i(t_i)<\frac{a}{t_i}\right]\geq1/4$ or not).
\end{observation}

We now prove that the expected cost paid by the algorithm in two consecutive phases is sufficiently large. Let the random variable $X_i=A_i+B_i$ denote the algorithm's cost in phase $i$, and recall that $X_i$ is determined by $y_0,\ldots,y_i$.

\begin{lemma}\label{lem_twophase}
For any $i$, and any $y_0,\ldots,y_{i-1}$, $\mathbb{E}_{y_i,y_{i+1}}[X_i+X_{i+1}]\geq\min((a\ln\rho)/32,1/64-o(1))$.
\end{lemma}

\begin{proof}
Given $y_0,\ldots,y_{i-1}$, if $\Pr_{y_i}\left[x_i(t_i)<\frac{a}{t_i}\right]<1/4$, then by Lemma \ref{lem_bad}, $\mathbb{E}_{y_i}[X_i]\geq\mathbb{E}_{y_i}[A_i]\geq(a\ln\rho)/8$. Otherwise, if $\Pr_{y_i}\left[x_i(t_i)<\frac{a}{t_i}\right]\geq1/4$, then by Lemma \ref{lem_good}, phase $i+1$ is well-started with probability at least $1/4$ over the choice of $y_i$. Thus, we have by Observation \ref{obs_wellstarted},
$\mathbb{E}_{y_i,y_{i+1}}[X_{i+1}]\geq\mathbb{E}_{y_i,y_{i+1}}[B_{i+1}]\geq\frac{1}{4}\times\min((a\ln\rho)/8,1/16-o(1))=\min((a\ln\rho)/32,1/64-o(1))$.
\end{proof}

Choose $r=\lfloor\sqrt{\ln n}/2\rfloor$, $\rho=e^r$, and $a=1/r$. Then $r=\omega(1)$, $\rho=\omega(r)$, and $\rho^{2r}=e^{2r^2}=e^{(\ln n)/2}=\sqrt{n}=o(n)$, as promised. With this, we are now set to prove the lower bound, which in turn, implies Theorem \ref{thm_lb} for MPMD.

\begin{theorem}\label{thm_lb_line}
Any randomized algorithm for MPMD must have a competitive ratio $\Omega(r)=\Omega(\sqrt{\log n})$ on the metric space of $n$ equispaced points in the unit interval, against an oblivious adversary.
\end{theorem}

\begin{proof}
With our choice of $a$ and $\rho$, the lower bound of Observation \ref{obs_wellstarted} becomes $1/16-o(1)$ and that of Lemma \ref{lem_twophase} becomes $1/64-o(1)$. Since phase $0$ is well-started, $\mathbb{E}_{y_0}[X_0]\geq1/16-o(1)$, by Observation \ref{obs_wellstarted}. Taking expectation over $y_1\ldots,y_r$, we have $\mathbb{E}[X_0]\geq1/16-o(1)$. Similarly, taking expectation of the result of Lemma \ref{lem_twophase} over $y_0,\ldots,y_{i-1},y_{i+2},\ldots,y_r$, we get $\mathbb{E}[X_i+X_{i+1}]\geq1/64-o(1)$, for every $i$. Summing up all of these, we infer that the algorithm's expected cost is at least $(1/128-o(1))r=\Omega(r)$. By Lemma \ref{lem_adv}, the optimum cost for any MPMD instance in the support of $\mathcal{D}$ is $2ar+1=3=O(1)$. Thus, the competitive ratio is $\Omega(r)=\Omega(\sqrt{\log n})$.
\end{proof}

\subsection{The $\Omega(\log^{1/3}n)$ Lower bound for MBPMD}

We give a distribution on input instances of MBPMD which ensures that any deterministic algorithm pays $\Omega(\log^{2/3}n)$ in expectation, while the instances have solutions of cost $O(\log^{1/3}n)$. 

As before, the construction of the bad distribution $\mathcal{D}$ on inputs depends on the three parameters: $r$, $\rho$, and $a$, which, in turn, depend on $n$. As before, we choose their values so that $r=\omega(1)$, $\rho=\omega(r)$, and $\rho^{2r}=o(n)$ (but the values are different from those in the MPMD construction, and we reveal them later).
The procedure for generating the distribution on input MPMBD instances is obtained by simply augmenting the procedure for MPMD from Section \ref{subsec_lb} with a rule to assign polarities to the requests. Nevertheless, we specify the entire procedure for the sake of completeness.
\begin{enumerate}
\item Initially, $n_0:=n$, where $n$ is even, and $S_0:=\{1/n,2/n,\ldots,1\}$.
\item For $i=0,\ldots,r$,
\begin{enumerate}
\item Sample $y_i$ from $U[0,1]$, the uniform distribution on the interval $[0,1]$.
\item $t_i:=\frac{a\rho^{1+y_i}}{n_i}$, $n_{i+1}:=2\left\lfloor\frac{a}{t_i}\right\rfloor=2\left\lfloor\frac{n_i}{\rho^{1+y_i}}\right\rfloor$.
\end{enumerate}
\item Choose an ``appropriate'' tuple $(s_0,\ldots,s_r)\in\{+1,-1\}^{r+1}$, whose existence is guaranteed by Lemma \ref{lem_polarity}.
\item For $i=0,\ldots,r$,
\begin{enumerate}
\item Give requests at points in $S_i$  with polarities alternating from left to right, with the request on the leftmost point in $S_i$ having polarity $s_i$.
\item Construct $S_{i+1}\subseteq S_i$ by scanning $S_i$ in ascending order, and including every $\left\lfloor\frac{n_i}{n_{i+1}}\right\rfloor^{\text{\tiny{th}}}$ point.
\item Wait for time $t_i$ (and then move on to the next phase, if $i<r$). 
\end{enumerate}
\end{enumerate}

Given $y_0,\ldots,y_r$ and $s_0,\ldots,s_r$, the locations and polarities of all the requests are determined. For $x\in[0,1]$, define the \textit{phase-$i$ cumulative surplus} at $x$ to be the signed total of the requests from phase $i$ that are located in $[0,x]$, and denote it by $\csur_i(x)$. Then $\csur_i(x)\in\{0,1\}$ if $s_i=+1$, and $\csur_i(x)\in\{-1,0\}$ if $s_i=-1$. Define $\csur(x)=\sum_{i=0}^r\csur_i(x)$, the cumulative surplus at $x$, which is the signed total of all requests from all phases that are located in $[0,x]$. Observe that for any $x$, any feasible solution to the instance must connect $|\csur(x)|$ requests located to the left of $x$ to the same number of requests located to the right of $x$. Hence, the connection cost of any feasible solution must be $\int_{x=0}^1|\csur(x)| dx$. Moreover, there exists a solution, say $\sol$, whose connection cost is precisely $\int_{x=0}^1|\csur(x)| dx$. This will be our adversarial solution to the instance. In order to bound the connection cost of $\sol$ from above, we prove that, having chosen $y_1,\ldots,y_r$, we can choose $s_1,\ldots,s_r$ so that $\int_{x=0}^1|\csur(x)| dx$ is small, and we use this choice of $s_1,\ldots,s_r$ in step 3 of the above procedure. Then we prove that the delay cost of $\sol$ is also sufficiently small, resulting in an upper bound on the cost of $\sol$.

\begin{lemma}\label{lem_polarity}
For every $y_0,\ldots,y_r$ there exists $(s_0,\ldots,s_r)\in\{+1,-1\}^{r+1}$ such that $\int_{x=0}^1|\csur(x)| dx=O(\sqrt{r})$.
\end{lemma}

\begin{proof}
We use the probabilistic method. We prove that if we choose $(s_0,\ldots,s_r)$ uniformly at random, then $\mathbb{E}\left[\int_{x=0}^1|\csur(x)| dx\right]=O(\sqrt{r})$. Since $\mathbb{E}\left[\int_{x=0}^1|\csur(x)| dx\right]=\int_{x=0}^1\mathbb{E}\left[|\csur(x)|\right] dx$, it is sufficient to prove that $\mathbb{E}\left[|\csur(x)|\right]=O(\sqrt{r})$.

Observe that $\csur_i(x)$ is zero if the number of requests of phase-$i$ in $[0,x]$ is even. If that number is odd, then $\csur_i(x)$ is $+1$ and $-1$ with probability $1/2$ each. Thus, $\csur(x)=\sum_{i=0}^r\csur_i(x)$ is the sum of at most $r+1$ independent random variables, each of which takes values $+1$ and $-1$ with equal probability, where the number of random variables is determined by $x$ and $y_0,\ldots,y_r$. Therefore $|\csur(x)|$ is the deviation of a random walk of at most $r+1$ steps on the integers starting from $0$, and moving in either direction with equal probability. Using a standard result\footnote{For instance: http://mathworld.wolfram.com/RandomWalk1-Dimensional.html}, we have $\mathbb{E}\left[|\csur(x)|\right]=O(\sqrt{r})$, as required.
\end{proof}

Taking the solution which minimizes the connection cost as the adversarial solution, we now prove an upper bound on the cost of the optimum solution of each instance in the support of the distribution generated by the adversarial procedure.

\begin{lemma}\label{lem_adv_b}
For any $y_0,\ldots,y_r$, the MBPMD instance generated has a solution of cost at most $2ar+O(\sqrt{r})+o(ar)$. 
\end{lemma}

\begin{proof}
Consider the solution $\sol$. By Lemma \ref{lem_polarity}, its connection cost can be made $O(\sqrt{r})$, with an appropriate choice of $(s_0,\ldots,s_r)$. Next, consider an arbitrary pair of requests in $\sol$. The waiting time of the earlier of the two requests in the pair is the difference between the arrival times of the requests, while the waiting time of the later request is zero. Thus, the delay cost of $\sol$ is bounded from above by the sum of the arrival times of all the requests, which is given by
\[\sum_{i=1}^rn_i\sum_{j=0}^{i-1}t_j=\sum_{i=1}^r\sum_{j=0}^{i-1}n_it_j\leq\sum_{i=1}^r\sum_{j=0}^{i-1}\frac{2an_i}{n_{j+1}}\leq2a\sum_{i=1}^r\sum_{j=0}^{i-1}\left(\frac{2}{\rho}\right)^{i-j-1}=2a\sum_{i=1}^r(1+o(1))=2a(r+o(r))\]
Summing the upper bounds on the connection and the delay cost of $\sol$, the claim stands proved.
\end{proof}

Next, in order to bound the expected cost of an arbitrary deterministic algorithm from below, we first state our choice of values of the parameters. This time we choose $r=\lfloor(\ln^{2/3}n)/4\rfloor$, $\rho=e^{\sqrt{r}}$, and $a=1/\sqrt{r}$. Then $r=\omega(1)$, $\rho=\omega(r)$, and $\rho^{2r}=e^{2r^{3/2}}\leq e^{(\ln n)/4}=n^{1/4}=o(n)$, as promised.

Observe that the distribution of the requests in space and time is identical to the distribution of requests in the random MPMD instance generated by the procedure from Section \ref{subsec_lb}, using the same values of the parameters. Thus, Lemmas \ref{lem_error}, \ref{lem_bad}, \ref{lem_good}, Observation \ref{obs_wellstarted}, and Lemma \ref{lem_twophase} all continue to hold, even if we allow the algorithm to connect requests of the same polarity. 

With this observation, we are ready to prove the lower bound on the competitive ratio, which, in turn, implies Theorem \ref{thm_lb} for MBPMD.

\begin{theorem}
Any randomized algorithm for MBPMD must have a competitive ratio $\Omega(\sqrt{r})=\Omega(\log^{1/3} n)$ on the metric space of $n$ equispaced points in the unit interval, against an oblivious adversary.
\end{theorem}

\begin{proof}
With our choice of the parameter values, the lower bound of Observation \ref{obs_wellstarted} becomes $1/16-o(1)$ and that of Lemma \ref{lem_twophase} becomes $1/64-o(1)$, as in the proof of Theorem \ref{thm_lb_line}, which again implies that the algorithm's expected cost is at least $(1/128-o(1))r=\Omega(r)$. By Lemma \ref{lem_adv_b}, the optimum cost for any MBPMD instance in the support of $\mathcal{D}$ is $2ar+O(\sqrt{r})+o(ar)=O(\sqrt{r})$ always. Thus, the competitive ratio is $\Omega(\sqrt{r})=\Omega(\log^{1/3}n)$.
\end{proof}

\section{Concluding Remarks and Open Problems}\label{sec_rem}

We improved the bounds on the competitive ratio of the problem of Min-cost Perfect Matching with Delays (MPMD), and also proved similar bounds for the bipartite variant of the problem. Our upper bound of $O(\log n)$ on $n$-point metric spaces proves that the aspect ratio of the underlying metric is not a blocker for competitiveness. On the other hand, our lower bounds are the first known lower bounds which increase with $n$. We mention here some of the missing pieces of the puzzle, and some extensions.

\begin{enumerate}
\item The immediate task is, arguably, to close the polylogarothmic gap between the upper and lower bounds for both MPMD and MBPMD. In order to improve our algorithm, we believe that it is necessary to bypass the tree-embedding step, since this step forces a distortion of $\Omega(\log n)$. 
\item Another follow-up task is to determine the competitiveness of deterministic algorithms for M(B)PMD on arbitrary metrics. It would not be surprising to discover an $\exp(\polylog(n))$ gap between the deterministic and the randomized bounds.
\item Finally, another interesting problem to pursue is the problem of min-cost $k$-dimensional matching, where the goal is to partition the requests into sets of size $k$. We need to identify interesting constraints on the connection cost, which generalize the metric property, and which admit a competitive algorithm.
\end{enumerate}

\section*{Acknowledgment}
The authors thank Amos Fiat for his insightful involvement in the discussions.

\bibliographystyle{plain}
\bibliography{references}

\begin{thebibliography}{10}

\bibitem{AggarwalGKM_SODA11}
Gagan Aggarwal, Gagan Goel, Chinmay Karande, and Aranyak Mehta.
\newblock Online vertex-weighted bipartite matching and single-bid budgeted
  allocations.
\newblock In {\em Proceedings of the Twenty-Second Annual {ACM-SIAM} Symposium
  on Discrete Algorithms}, pages 1253--1264, 2011.

\bibitem{BahmaniK_ESA10}
Bahman Bahmani and Michael Kapralov.
\newblock Improved bounds for online stochastic matching.
\newblock In {\em Algorithms - {ESA} 2010, 18th Annual European Symposium},
  pages 170--181, 2010.

\bibitem{BansalBGN_Algorithmica14}
Nikhil Bansal, Niv Buchbinder, Anupam Gupta, and Joseph Naor.
\newblock A randomized o(log2 k)-competitive algorithm for metric bipartite
  matching.
\newblock {\em Algorithmica}, 68(2):390--403, 2014.

\bibitem{BansalBMN_JACM15}
Nikhil Bansal, Niv Buchbinder, Aleksander Madry, and Joseph Naor.
\newblock A polylogarithmic-competitive algorithm for the \emph{k}-server
  problem.
\newblock {\em J. {ACM}}, 62(5):40, 2015.

\bibitem{Bartal_FOCS96}
Yair Bartal.
\newblock Probabilistic approximations of metric spaces and its algorithmic
  applications.
\newblock In {\em 37th Annual Symposium on Foundations of Computer Science},
  pages 184--193, 1996.

\bibitem{BorodinE}
Allan Borodin and Ran El{-}Yaniv.
\newblock {\em Online computation and competitive analysis}.
\newblock Cambridge University Press, 1998.

\bibitem{BorodinE_InfComput99}
Allan Borodin and Ran El{-}Yaniv.
\newblock On randomization in on-line computation.
\newblock {\em Inf. Comput.}, 150(2):244--267, 1999.

\bibitem{DevanurJ_STOC12}
Nikhil~R. Devanur and Kamal Jain.
\newblock Online matching with concave returns.
\newblock In {\em Proceedings of the 44th Symposium on Theory of Computing
  Conference}, pages 137--144, 2012.

\bibitem{DevanurJK_SODA13}
Nikhil~R. Devanur, Kamal Jain, and Robert~D. Kleinberg.
\newblock Randomized primal-dual analysis of {RANKING} for online bipartite
  matching.
\newblock In {\em Proceedings of the Twenty-Fourth Annual {ACM-SIAM} Symposium
  on Discrete Algorithms}, pages 101--107, 2013.

\bibitem{Edmonds_JRNBS65}
Jack Edmonds.
\newblock Maximum matching and a polyhedron with o,1-vertices.
\newblock {\em Journal of Research of the National Bureau of Standards},
  69B:125--130, 1965.

\bibitem{Edmonds_CJM65}
Jack Edmonds.
\newblock Paths, trees, and flowers.
\newblock {\em Canadian Journal of Mathematics}, 17:449--467, 1965.

\bibitem{EmekKW_STOC16}
Yuval Emek, Shay Kutten, and Roger Wattenhofer.
\newblock Online matching: haste makes waste!
\newblock In {\em Proceedings of the 48th Annual {ACM} {SIGACT} Symposium on
  Theory of Computing}, pages 333--344, 2016.

\bibitem{EpsteinLSW_STACS13}
Leah Epstein, Asaf Levin, Danny Segev, and Oren Weimann.
\newblock Improved bounds for online preemptive matching.
\newblock In {\em 30th International Symposium on Theoretical Aspects of
  Computer Science}, pages 389--399, 2013.

\bibitem{FakcharoenpholRT_JCSS04}
Jittat Fakcharoenphol, Satish Rao, and Kunal Talwar.
\newblock A tight bound on approximating arbitrary metrics by tree metrics.
\newblock {\em J. Comput. Syst. Sci.}, 69(3):485--497, 2004.

\bibitem{FeldmanMMM_FOCS09}
Jon Feldman, Aranyak Mehta, Vahab~S. Mirrokni, and S.~Muthukrishnan.
\newblock Online stochastic matching: Beating 1-1/e.
\newblock In {\em 50th Annual {IEEE} Symposium on Foundations of Computer
  Science}, pages 117--126, 2009.

\bibitem{GuptaL_ICALP12}
Anupam Gupta and Kevin Lewi.
\newblock The online metric matching problem for doubling metrics.
\newblock In {\em Automata, Languages, and Programming - 39th International
  Colloquium}, pages 424--435, 2012.

\bibitem{KalyanasundaramP_JAlgorithms93}
Bala Kalyanasundaram and Kirk Pruhs.
\newblock Online weighted matching.
\newblock {\em J. Algorithms}, 14(3):478--488, 1993.

\bibitem{KarpVV_STOC90}
Richard~M. Karp, Umesh~V. Vazirani, and Vijay~V. Vazirani.
\newblock An optimal algorithm for on-line bipartite matching.
\newblock In {\em Proceedings of the 22nd Annual {ACM} Symposium on Theory of
  Computing}, pages 352--358, 1990.

\bibitem{KhullerMV_TheorComputSci94}
Samir Khuller, Stephen~G. Mitchell, and Vijay~V. Vazirani.
\newblock On-line algorithms for weighted bipartite matching and stable
  marriages.
\newblock {\em Theor. Comput. Sci.}, 127(2):255--267, 1994.

\bibitem{KoutsoupiasN_WAOA03}
Elias Koutsoupias and Akash Nanavati.
\newblock The online matching problem on a line.
\newblock In {\em Approximation and Online Algorithms, First International
  Workshop}, pages 179--191, 2003.

\bibitem{MahdianY_STOC11}
Mohammad Mahdian and Qiqi Yan.
\newblock Online bipartite matching with random arrivals: an approach based on
  strongly factor-revealing lps.
\newblock In {\em Proceedings of the 43rd {ACM} Symposium on Theory of
  Computing}, pages 597--606, 2011.

\bibitem{McGregor_Approx05}
Andrew McGregor.
\newblock Finding graph matchings in data streams.
\newblock In {\em 8th International Workshop on Approximation Algorithms for
  Combinatorial Optimization Problems}, pages 170--181, 2005.

\bibitem{MeyersonNP_SODA06}
Adam Meyerson, Akash Nanavati, and Laura~J. Poplawski.
\newblock Randomized online algorithms for minimum metric bipartite matching.
\newblock In {\em Proceedings of the Seventeenth Annual {ACM-SIAM} Symposium on
  Discrete Algorithms}, pages 954--959, 2006.

\bibitem{StougieV_OperResLett02}
Leen Stougie and Arjen P.~A. Vestjens.
\newblock Randomized algorithms for on-line scheduling problems: how low can't
  you go?
\newblock {\em Oper. Res. Lett.}, 30(2):89--96, 2002.

\bibitem{Varadaraja_ICALP11}
Ashwinkumar~Badanidiyuru Varadaraja.
\newblock Buyback problem - approximate matroid intersection with cancellation
  costs.
\newblock In {\em Automata, Languages and Programming - 38th International
  Colloquium}, pages 379--390, 2011.

\bibitem{Yao_FOCS77}
Andrew~Chi{-}Chih Yao.
\newblock Probabilistic computations: Toward a unified measure of complexity
  (extended abstract).
\newblock In {\em 18th Annual Symposium on Foundations of Computer Science},
  pages 222--227, 1977.

\end{thebibliography}

\appendix
\section{Appendix: Proofs missing from Section \ref{subsec_reduction}}\label{app_a}

\begin{proof}[Proof of Lemma \ref{lem_embed}]
The result by Fakcharoenphol et al.\ (Theorem 2 of \cite{FakcharoenpholRT_JCSS04}) states that $\mathcal{M}$ can be embedded into a distribution over $2$-HSTs with distortion $O(\log n)$, where the points of $\mathcal{M}$ are the leaves of the trees in the support of the distribution. The result by Bansal et al.\ (Theorem 8 of \cite{BansalBMN_JACM15}) states that a $\sigma$-HST with leaf set $S$ of size $n$ can be deterministically embedded into a metric given by a tree of height $O(\log n)$ with the same leaf set, such that the distance between any pair of leaves is distorted by at most $2\sigma/(\sigma-1)$. Composing these embeddings, we get the claimed result.
\end{proof}

\begin{proof}[Proof of Lemma \ref{lem_reduction}]
Algorithm $\mathcal{A}$ simply samples a metric space $\mathcal{M}'=(S',d')$ from the distribution $\mathcal{D}$, and simulates the behavior of $\mathcal{A}^{\mathcal{M}'}$. Fix an input instance $I$ of M(B)PMD on $\mathcal{M}$, and an arbitrary solution $\sol$ of $I$. Let $\alg$ be the solution output by $\mathcal{A}^{\mathcal{M}'}$, and hence, by $\mathcal{A}$. However, note that the costs paid by $\mathcal{A}^{\mathcal{M}'}$ and $\mathcal{A}$ are different, since they are working on different metrics. Since $d(p_1,p_2)\leq d'(p_1,p_2)$, for all $p_1,p_2\in S$, 
\[\mathcal{A}(I)=\alg_d+\alg_t\leq\alg_{d'}+\alg_t=\mathcal{A}^{\mathcal{M}'}(I)\]
Since $\mathcal{A}^{\mathcal{M}'}$ is $(\beta,\gamma)$-competitive, we have by definition,
\[\mathcal{A}^{\mathcal{M}'}(I)\leq\beta\sol_{d'}+\gamma\sol_t\]
Thus, $\mathcal{A}(I)\leq\beta\sol_{d'}+\gamma\sol_t$. Taking expectation over the random choice of $\mathcal{M'}$,
\[\mathbb{E}_{\mathcal{M}'=(S',d')\sim\mathcal{D}}[\mathcal{A}(I)]\leq\beta\cdot\mathbb{E}_{\mathcal{M}'=(S',d')\sim\mathcal{D}}[\sol_{d'}]+\gamma\cdot\mathbb{E}_{\mathcal{M}'=(S',d')\sim\mathcal{D}}[\sol_t]\]
Using linearity of expectation, the fact that the embedding has distortion $\mu$, and that $\sol_t$ is independent of $\mathcal{M}'$, we have,
\[\mathbb{E}_{\mathcal{M}'=(S',d')\sim\mathcal{D}}[\mathcal{A}(I)]\leq\beta\mu\cdot\sol_{d}+\gamma\cdot\sol_t\]
Thus, the claim follows from the definition of $(\beta,\gamma)$-competitiveness.
\end{proof}

\section{Appendix: Proofs missing from Section \ref{subsec_lb}}\label{app_b}

\begin{proof}[Proof of Lemma \ref{lem_error}]
The points in $S_{i+1}$ are spaced at least a distance $d_{i+1}$ apart, where $d_0=1/n$, and
\[d_{i+1}=d_i\times\left\lfloor\frac{n_i}{n_{i+1}}\right\rfloor\geq d_i\times\left(\frac{n_i}{n_{i+1}}-1\right)\]
As a result, for all $i=0,\ldots,r-1$,
\[n_{i+1}d_{i+1}\geq n_{i+1}d_i\left(\frac{n_i}{n_{i+1}}-1\right)=n_id_i\left(1-\frac{n_{i+1}}{n_i}\right)\geq n_id_i\left(1-\frac{2}{\rho}\right)\]
Thus, by induction, for all $i$, $n_id_i\geq (1-2/\rho)^in_0d_0\geq1-2i/\rho\geq1-2r/\rho=1-o(1)$, since $\rho=\omega(r)$.
\end{proof}


\begin{proof}[Proof of Lemma \ref{lem_bad}]
Observe that given $y_0,\ldots y_{i-1},y_i$,
\[A_i=\int_0^{t_i}x_i(t) dt=\int_0^{\frac{a\rho^{1+y_i}}{n_i}}x_i(t) dt\]
Given $y_0,\ldots,y_{i-1}$, the expectation of the above over $y_i$ is
\[\mathbb{E}_{y_i}[A_i]=\int_{y=0}^1\int_{t=0}^{\frac{a\rho^{1+y}}{n_i}}x_i(t) dt dy \geq\int_{y=0}^1\int_{t=\frac{a\rho}{n_i}}^{\frac{a\rho^{1+y}}{n_i}}x_i(t) dt dy =\int_{t=\frac{a\rho}{n_i}}^{\frac{a\rho^2}{n_i}}\int_{y=\frac{\ln\left(\frac{n_it}{a\rho}\right)}{\ln\rho}}^1x_i(t) dy dt\]
Thus,
\[\mathbb{E}_{y_i}[A_i]\geq\int_{t=\frac{a\rho}{n_i}}^{\frac{a\rho^2}{n_i}}x_i(t)\left[\int_{y=\frac{\ln\left(\frac{n_it}{a\rho}\right)}{\ln\rho}}^1 dy\right] dt =\int_{t=\frac{a\rho}{n_i}}^{\frac{a\rho^2}{n_i}}\left(1-\frac{\ln\left(\frac{n_it}{a\rho}\right)}{\ln\rho}\right)x_i(t) dt\geq\int_{t=\frac{a\rho}{n_i}}^{\frac{a\rho^{3/2}}{n_i}}\frac{1}{2}\cdot x_i(t) dt\]
The last inequality holds because the factor $\left(1-\frac{\ln\left(\frac{n_it}{a\rho}\right)}{\ln\rho}\right)$ in the integrand is at least $1/2$ for $t\in\left[\frac{a\rho}{n_i},\frac{a\rho^{3/2}}{n_i}\right]$, and non-negative for $t\in\left[\frac{a\rho^{3/2}}{n_i},\frac{a\rho^2}{n_i}\right]$.
Using $\mathbb{I}[\cdot]$ to denote the indicator function, we have
\[\mathbb{E}_{y_i}[A_i]\geq\frac{1}{2}\int_{t=\frac{a\rho}{n_i}}^{\frac{a\rho^{3/2}}{n_i}}x_i(t) dt \geq\frac{1}{2}\int_{t=\frac{a\rho}{n_i}}^{\frac{a\rho^{3/2}}{n_i}}x_i(t)\cdot\mathbb{I}\left[x_i(t)\geq\frac{a}{t}\right] dt \geq\frac{1}{2}\int_{t=\frac{a\rho}{n_i}}^{\frac{a\rho^{3/2}}{n_i}}\frac{a}{t}\cdot\mathbb{I}\left[x_i(t)\geq\frac{a}{t}\right] dt\]
The last inequality follows from the fact that if $x_i(t)<\frac{a}{t}$, then $\mathbb{I}\left[x_i(t)\geq\frac{a}{t}\right]=0$. 
Introducing the change of variables $t=t(z)=\frac{a\rho^{1+z}}{n_i}$, we have $dt=t(z)\ln\rho\cdot dz$, and $z$ increases from $0$ to $1/2$ as $t$ increases from $\frac{a\rho}{n_i}$ to $\frac{a\rho^{3/2}}{n_i}$. Thus,
\[\mathbb{E}_{y_i}[A_i]\geq\frac{1}{2}\int_{z=0}^{\frac{1}{2}}\frac{a}{t(z)}\cdot\mathbb{I}\left[x_i(t(z))\geq\frac{a}{t(z)}\right]t(z)\ln\rho\cdot dz=\frac{a\ln\rho}{2}\int_{z=0}^{\frac{1}{2}}\mathbb{I}\left[x_i(t(z))\geq\frac{a}{t(z)}\right] dz\]
Observe that the value of the integral is precisely $\Pr_{y_i}\left[x_i(t_i)\geq\frac{a}{t_i}\text{ and }y_i\leq 1/2\right]$. We are given that $\Pr_{y_i}\left[x_i(t_i)<\frac{a}{t_i}\right]<1/4$. Since $y_i$ is uniform in $[0,1]$, by the union bound, we get
\[\Pr_{y_i}\left[x_i(t_i)\geq\frac{a}{t_i}\text{ and }y_i\leq 1/2\right]>\frac{1}{4}\]
Substituting this, we get $\mathbb{E}_{y_i}[A_i]\geq(a\ln\rho)/8$.
\end{proof}

\end{document}